\documentclass[11pt]{article}

\usepackage{a4wide}
\usepackage{latexsym}
\usepackage{amssymb}
\usepackage{amsmath}
\usepackage{epsfig}
\usepackage{subfigure}
\usepackage{url}
\usepackage{fullpage}
\usepackage[labelsep=none]{caption}

\newtheorem{theorem}{Theorem}

\newtheorem{lemma}[theorem]{Lemma}

\newcommand\eps{\varepsilon}

\DeclareMathOperator{\bob}{{\sf Bob}}
\DeclareMathOperator{\stairs}{{\sf stairs}}

\newcommand{\qed}{\hspace{\stretch{1}}$\Box$}
\newenvironment{proof}{\vspace{-.25\baselineskip}\noindent\textbf{Proof.}
}{\qed\par\medskip}

\begin{document}
%
\title
{A Doubly Exponentially Crumbled Cake}

\author{
{Tobias Christ \footnote{Institute of Theoretical Computer Science, ETH Z\"urich, 8092 Z\"urich, Switzerland, {\texttt{\{christt, afrancke, gebauerh\}@inf.ethz.ch}} } }\quad
{Andrea Francke $^*$} \quad
{Heidi Gebauer $^*$} \quad
{Ji\v{r}\'{\i} Matou\v{s}ek \footnote{Dept. of Applied Mathematics and Institute of Theoretical Computer Science, Charles University, Malostransk\'{e} n\'{a}m. 25,
118~00~~Praha~1, Czech Republic, and Institute of Theoretical Computer Science, ETH Z\"urich, 8092 Z\"urich, Switzerland, {\texttt{matousek@kam.mff.cuni.cz}}} } \quad
{Takeaki Uno \footnote{National Institute of Informatics, 2-1-2, Hitotsubashi, Chiyoda-ku,
Tokyo 101-8430, Japan, {\texttt{uno@nii.jp}}}}
}

\maketitle
%
\begin{abstract}
We consider the following cake cutting game: 
Alice chooses a set~$P$ of $n$~points in
the square (cake)~$[0,1]^2$, where $(0,0) \in P$; 
Bob cuts out $n$ axis-parallel rectangles with disjoint
interiors, each of them having a point of $P$ as the
lower left corner;  Alice keeps the rest.
It has been conjectured that Bob can always secure at least half
of the cake. This remains unsettled, and it is not even known
whether Bob can get any positive fraction independent of~$n$.
 We prove that \emph{if} Alice can force Bob's share
to tend to zero, \emph{then} she must use very many points; namely,
to prevent Bob from gaining more than $1/r$ of the cake,
she needs at least $2^{2^{\Omega(r)}}$ points. 
\end{abstract}
%
%
%
%
\section{Introduction}

Alice has baked a square cake with raisins for Bob, but
really she would like to keep most of it for herself.
In this, she relies on a peculiar habit of Bob: he eats only
rectangular pieces of the cake, with sides parallel
to the sides of the cake, that contain exactly one raisin each,
and that raisin has to be exactly in the lower left corner
(see Fig.~\ref{f:example}). Alice gets whatever remains
after Bob has cut out all such pieces. In order to give
Bob at least some chance, Alice has to put a raisin
in the lower left corner of the whole cake. 

Mathematically, the cake is the square $[0,1]^2$, the raisins
form an $n$-point set $P\subset [0,1]^2$, where
$(0,0)\in P$ is required, and Bob's share consists of 
$n$ axis-parallel rectangles with disjoint
interiors, each of them having a point of $P$ as the
lower left corner.

By placing points densely along the main diagonal,
Alice can limit Bob's share to~$\frac 12+\eps$,
with $\eps>0$ arbitrarily small.
A natural question then is, can Bob always obtain
 at least half of the cake?

This question (in a cake-free formulation) appears in
Winkler~\cite{Win07} (``Packing Rectangles'', p.~133),
where he claims it to be at least 10 years old and of
origin unknown to him. The first written reference seems
to be an IBM puzzle webpage~\cite{IBM04}.

\begin{figure}[tbh]
\centering
\subfigure{
\includegraphics[width=0.3\textwidth]{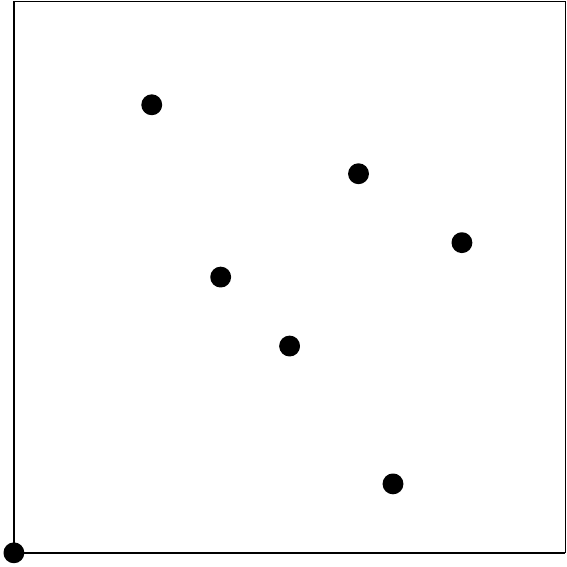} }
~~~~~~~~~~~~~~~~~~
\subfigure{
\includegraphics[width=0.3\textwidth]{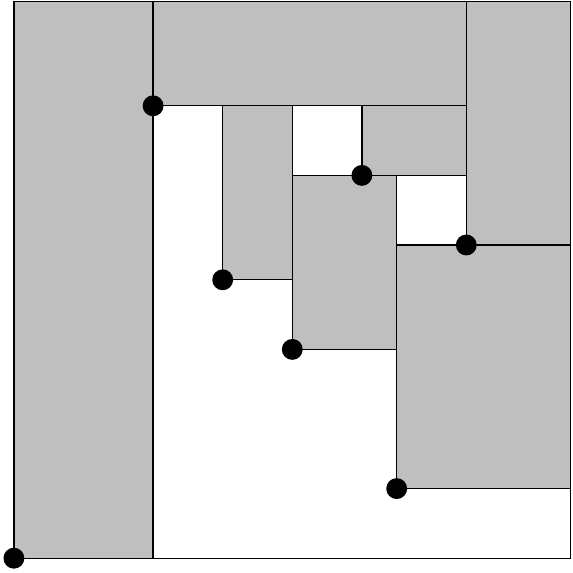}}
\caption{\label{f:example} Example: Alice's points (left)
and Bob's rectangles (right).}
\end{figure}

We tried to answer the question and could not, probably similar to many 
other people before us. We believe that there are no simple examples
leaving more than $\frac 12$ to Alice, but on the other hand,
it seems difficult to prove even that Bob can always secure
$0.0001\%$ of the cake. We were thus led to seriously considering
the possibility that Alice might be able
to limit Bob's share to less than $1/r$, for every $r>0$,
but that the number of points $n$ she would need 
would grow enormously as a function of~$r$.

Here we prove a doubly exponential lower bound on this function.
First we introduce the following notation. For a finite
$P\subset[0,1]^2$,
let $\bob(P)$ be the largest area Bob can win for $P$, and
let $\bob(n)$ be the infimum of $\bob(P)$ over all $n$-point $P$
as above.\footnote{It is easily checked that, given $P$,
there are finitely many possible placement of 
Bob's \emph{inclusion-maximal} rectangles, and therefore,
$\bob(P)$ is attained by some choice of rectangles.
On the other hand, it is not so clear whether
$\bob(n)$ is attained; we leave this question
aside.} Also, for a real number $r>1$ let
$n(r):=\min\{n: \bob(n)\le 1/r\} \in \{1,2,\ldots\}\cup \{\infty\}$.

\begin{theorem}\label{t:} There exists a constant $r_0$ such that for all $r\ge r_0$,
$
n(r)\ge 2^{2^{r/2}}$.
\end{theorem}

The only previous results on this problem we could find
is the Master's thesis of M\"{u}ller-Itten \cite{Mue10}. She conjectured
that  Alice's optimal strategy is placing the $n$ points
on the main diagonal with equal spacing (for which Bob's share
is $\frac{1}{2}\left(1 + \frac{1}{n}\right)$). She proved
this conjecture for $n\le 4$, and also in the ``grid''
case with $P=\{(0,0)\}\cup \{ (\frac{i}{n}, \frac{\pi(i)}{n}):i \in \{1, \ldots, n-1\}\}$, where $\pi$ is a permutation of $\{1, \ldots, n-1\}$.
She also showed that $\bob(n)\ge\frac 1n$.

The problem considered here can be put into a wider context.
Various problems of fair division of resources, often phrased
as cake-cutting problems, go back at least to Steinhaus, Banach and Knaster;
see, e.g., \cite{RW98}. Even closer to our particular setting
is Winkler's \emph{pizza problem}, recently solved by 
Cibulka et al.~\cite{CKMS10}.  
%
%

\section{Preliminaries} 
We call a point~$a$ a \emph{minimum} of a set $X\subseteq [0,1]^2$
if there is no $b\in X\setminus \{a\}$ for which both 
$x(b)\le x(a)$ and $y(b)\le y(a)$. 
Let $p_1,p_2,\ldots,p_k$ be an enumeration
of the minima of $P\setminus \{(0,0)\}$ in the order
of decreasing $y$-coordinate (and increasing $x$-coordinate).
Let $\stairs(P)$ be the union of all the axis-parallel rectangles
with lower left corners at $(0,0)$ whose interior avoids $P$;
see Fig.~\ref{f:sta}(a).
%
%

\begin{figure}[tbh]
\centering
\subfigure[$\stairs(P)$ and the subproblems.]
{\includegraphics[width=0.35\textwidth]
{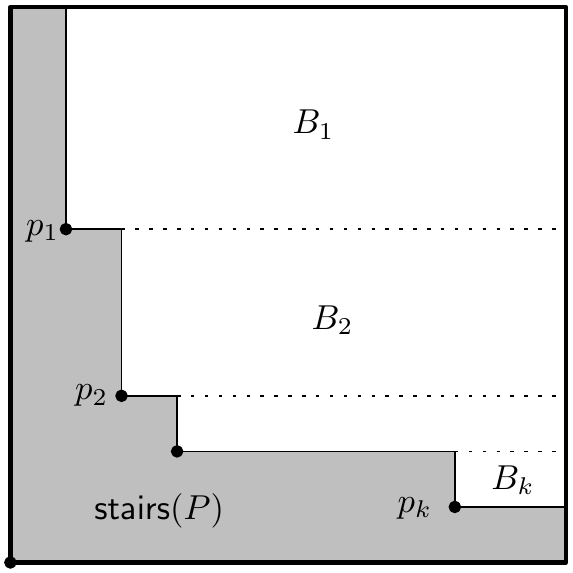}}
\qquad
\subfigure[Illustration to the proof of Lemma~\ref{l:nobig}.]{
\includegraphics[width=0.41\textwidth]
{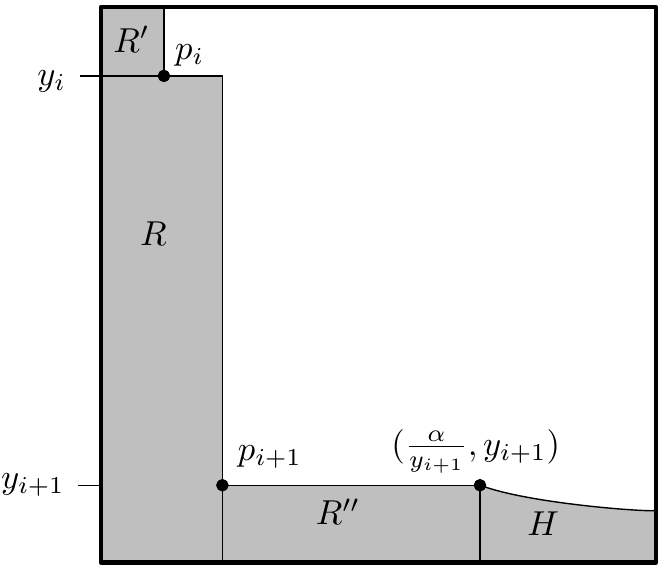}}
\caption{\label{f:sta}}
\end{figure}


Furthermore, let $s$ be the area of $\stairs(P)$, and let $\alpha$ be the largest
area of an axis-parallel rectangle contained in $\stairs(P)$.
Let us also define $\rho :=\frac s\alpha$.
For a point $p\in P$ and an axis-parallel rectangle $B\subseteq[0,1]^2$
 with lower left corner at $p$, we denote by~$a$ be the maximum area
of the cake Bob can gain in~$B$ using only rectangles
with lower left corner in points of $B\cap P$.
By re-scaling, we have $a=\beta\cdot\bob(P_B)$,
where $\beta$ is the area of~$B$ and $P_B$ denotes the set $P\cap B$
transformed by the affine transform that maps $B$ onto $[0,1]^2$.

We will use the monotonicity of $\bob(\cdot)$, i.e.,
$\bob(n+1)\le\bob(n)$ for all $n\ge 1$.
Indeed, Alice can always place an extra point on the right side
of the square, say, which does not influence Bob's share.

\section{The decomposition} We decompose the complement of $\stairs(P)$
into horizontal rectangles $B_1,\ldots,B_k$ as indicated
in Fig.~\ref{f:sta}(a), so that $p_i$ is the lower left
corner of $B_i$. Let $\beta_i$ be the area of $B_i$;
we have $s+\sum_{i=1}^k\beta_i=1$.

By the above and by an obvious superadditivity, we have
\begin{equation}\label{e:decompose}
\bob(P)\ge \alpha+\sum_{i=1}^k \beta_i \bob(P_i),
\end{equation}
where $P_i:= P_{B_i}.$ (This is a somewhat simple-minded estimate, since it doesn't take
into account any interaction among the $B_i$).

The following lemma captures the main properties of this
decomposition.

\begin{lemma}\label{l:nobig}
Let us  assume that $\rho=\frac s\alpha\ge r_0$, where
$r_0$ is a suitable (sufficiently large) constant.
Then 
\begin{itemize}\item 
$s \le \frac 14 \cdot 2^{-\rho}$ (the staircase has a small area), and
\item 
$\sum_{j:j\ne i}\beta_i \ge  2^\rho s$ 
for every $i=1,2,\ldots,k$ (none of the subproblems
occupies almost all of the area).
\end{itemize}
\end{lemma}

\begin{proof} First we note that 
since no rectangle with lower left corner $(0,0)$
and upper right corner in $\stairs(P)$ has area bigger than
$\alpha$, the region $\stairs(P)$ lies below the hyperbola
$y=\frac \alpha x$. Thus
$s\le \alpha+\int_{\alpha}^1\frac \alpha x\,{\rm d}x=\alpha+\alpha
\ln\frac 1\alpha$. This yields $\alpha\le e^{-\rho+1}$,
and so $s=\rho\alpha \le \rho e^{-\rho+1}\le \frac 14 \cdot 2^{-\rho}$
(for $\rho$ sufficiently large).

It remains to show that $\sum_{j:j\ne i}\beta_i \ge 2^\rho s$;
since $\sum_{j=1}^k\beta_j=1-s$, it suffices to show 
$\beta_i\le 1-2\cdot 2^\rho s$ for all~$i$.
 
Let $y_i$ be the $y$-coordinate of $p_i$ for $i \geq 1$, and let $y_{0} = 1$;
we have $\beta_{i + 1} \le y_{i}-y_{i+1}$ for $i \geq 0$. 

First, if $y_i\le\frac 12$, then $\beta_{i+1}\le
\frac 12\le 1-2\cdot 2^\rho s$ by the above, and 
we are done. So we assume $y_i>\frac 12$.

The area of $\stairs(P)$ can be bounded from above as indicated
in Fig.~\ref{f:sta}(b). Namely, the rectangle $R$ has area 
at most $\alpha$ (since it is contained in $\stairs(P)$), and
the rectangle $R'$ above it also has area no more than $\alpha$ 
(using $y_i>\frac 12$).  The top right corner of $R''$
lies on the hyperbola $y=\frac\alpha x$ used above, and thus $R''$
has area at most $\alpha$ as well. Finally, the region $H$
on the right of $R''$ and below the hyperbola has area
$\int_{\alpha/y_{i+1}}^1\frac\alpha x\,{\rm d}x =\alpha\ln(y_{i+1}/\alpha)$.



Since $\stairs(P)\subseteq R\cup R'\cup R''\cup H$, we 
have $s\le \alpha(3+\ln(y_{i+1}/\alpha))$. Using $\rho=\frac s\alpha$
we  obtain $y_{i+1}\ge \alpha e^{\rho-3}= se^{\rho-3}/\rho
\ge 2\cdot2^{\rho} s$ (again using the assumption that $\rho$ is large).

Finally, we have $\beta_{i+1}\le 1-y_{i+1}\le 1-2\cdot 2^\rho s$,
and the lemma is proved.
\end{proof}

\section{Proof of Theorem~\ref{t:}}
\begin{proof}
Let $r\ge r_0$.
We may assume that $r$ is of the form $r=1/\bob(n)$,
where $n=n(r)$. In particular, $\bob(m)>\frac 1r$ for all $m<n$.

We will derive the following recurrence for such an $r$:
\begin{equation}\label{e:recur}
n(r)\ge 2 n(r-2^{-(r+1)/2}).
\end{equation}
Applying it iteratively $t:=2^{r/2}$ times, we find that
$n(r)\ge 2^t n(r-1)\ge 2^t$ as claimed in the theorem.

We thus start with the derivation of (\ref{e:recur}).
Let us look at the inequality (\ref{e:decompose}) for an $n$-point set $P$
that attains $\bob(n)$.\footnote{Or rather,
since we haven't proved that $\bob(n)$
is attained, we should choose $n$-point $P$
with  $\bob(P)<\bob(n')$ for all $n'<n$.}
Since $n_i:=|P_i|<n$ for all $i$, we have
$\bob(P_i)>\frac 1r$ for all $i$. 

Let $\alpha$ and $s$ be as above.
First we derive $\rho=\frac s\alpha\ge r$.
Indeed, if we had $\alpha>\frac sr$, then the right-hand of 
(\ref{e:decompose}) can be estimated as follows:
$$
\alpha+\sum_{i=1}^k \beta_i \bob(P_i)> 
\frac 1r\biggl(s+\sum_{i=1}^k\beta_i\biggr)=
\frac1r,
$$
which contradicts the inequality~(\ref{e:decompose}).
So $\rho\ge r\ge r_0$ indeed.

Let us set $\gamma_i :=
\bob(P_i)-\frac 1r$; this is Bob's ``gain'' over the ratio $\frac 1r$
in the $i$th subproblem. From (\ref{e:decompose}) we have
\begin{eqnarray*}
\frac 1r&\ge& \sum_{i=1}^k\beta_i\left(\frac 1r+\gamma_i\right)
\ge \frac 1r\biggl( \sum_{i=1}^k\beta_i\biggr)+\sum_{i=1}^k\beta_i\gamma_i
=\frac {1-s}r +\sum_{i=1}^k\beta_i\gamma_i,
\end{eqnarray*}
and so
\begin{equation}\label{e:gains}
\sum_{i=1}^k\beta_i\gamma_i\le \frac sr.
\end{equation}

According to Lemma~\ref{l:nobig}, we can partition the index set
$\{1,2,\ldots,k\}$ into two subsets $I_1,I_2$ so that
$\sum_{i\in I_j}\beta_i\ge 2^\rho s \ge 2^r s$ for $j=1,2$.

Let $i_1$ be such that $\gamma_{i_1}=\min_{i\in I_1}\gamma_i$, and
similarly for $i_2$. Then (\ref{e:gains}) gives, for $j=1,2$,
$$
\frac sr \ge \sum_{i\in I_j}\beta_i\gamma_i\ge \gamma_{i_j}\sum_{i\in I_j}\beta_i\ge \gamma_{i_j} 2^r s,
$$
and so $\gamma_{i_j}\le \gamma^*:= 2^{-r}/r$.

Let us define $r^*<r$ by $\frac 1{r^*}=\frac1r+\gamma^*$.
Then we know that at least two of the sets $P_i$ contain at least
$n(r^*)$ points each, and hence $n(r)\ge 2 n(r^*)$.
We calculate $r^*=\frac r{1+r\gamma^*}\ge r(1-r\gamma^*)= r-r2^{-r}\ge
r-2^{-(r+1)/2}$ (again using $r\ge r_0$).

So we have derived the desired recurrence (\ref{e:recur}),
and Theorem~\ref{t:} is proved.
\end{proof}
\subsection*{Acknowledgments}
This research was partially done at the \emph{Gremo Workshop on Open Problems} 2010,
and the support of the ETH Z\"urich is gratefully acknowledged.
We would like to thank Michael Hoffmann and Bettina Speckmann
for useful discussion, and the organizers and participants of GWOP~2010 for a beautiful workshop.

%
%

\bibliographystyle{abbrv}

\end{document}